\documentclass[runningheads,a4paper]{llncs}

\usepackage{fullpage}
\usepackage{amssymb}		% useful mathematical symbols
\usepackage{algorithm}
\usepackage{algcompatible}

\newcommand{\fa}{\forall}
\newcommand{\ex}{\exists}
\newcommand{\bac}{\backslash}
\newcommand{\seb} {\subseteq}
\newcommand{\sepq}{\supseteq} % \sep ya se usa en keywords
\newcommand{\matN}{\mathbb N}
\newcommand{\matR}{\mathbb R}

\newcommand{\cG}{{\cal G}}
\newcommand{\cP}{{\cal P}}
\newcommand{\cH}{{\cal H}}
\newcommand{\cJ}{{\cal J}}
\newcommand{\cK}{{\cal K}}
\newcommand{\cL}{{\cal L}}
\newcommand{\cW}{{\cal W}}

\begin{document}

\mainmatter

\title{On the Complexity of the Decisive Problem\\in Simple, Regular and Weighted Games}

\author{Andreas Polym\'eris\inst{1} \and Fabi\'an Riquelme\inst{2}\thanks{This author was supported by grant BecasChile of the ``National Commission for Scientific and Technological Research of Chile'' (CONICYT) of the Chilean Government.}}
\institute{Dept. de Ingenier\'ia Inform\'atica y Cs. de la Computaci\'on, Universidad de Concepci\'on, Concepci\'on, Chile.\\
	\and Dept. de Llenguatges i Sistemes Inform\`atics, Universitat Polit\`ecnica de Catalunya, Barcelona, Spain. \\
	\email{apolymer@udec.cl, farisori@lsi.upc.edu}
}
\date{}

\maketitle
%\author[Pol]{Andreas Polym\'eris}
%\ead{apolymer@udec.cl}
%\author[Riq]{Fabi\'an Riquelme\fnref{fn1}}
%\ead{farisori@lsi.upc.edu}
%\fntext[fn1]{This author was supported by grant BecasChile of the ``National Commission for Scientific and Technological Research of Chile'' (CONICYT) of the Chilean Government.}

%\address[Pol]{Departamento de Ingenier\'ia Inform\'atica y Ciencias de la Computaci\'on, Universidad de Concepci\'on, Concepci\'on, Chile}
%\address[Riq]{Departament de Llenguatges i Sistemes Inform\`atics, Universitat Polit\`ecnica de Catalunya, Barcelona, Spain}

\begin{abstract}
We study the computational complexity of an important property of simple, regular and weighted games, which is decisiveness. 
We show that this concept can naturally be represented in the context of hypergraph theory, 
and that decisiveness can be decided for simple games in quasi-polynomial time, and for regular and weighted games in polynomial time.
The strongness condition poses the main difficulties, while properness reduces the complexity of the problem, especially if it is amplified by regularity.
On the other hand, regularity also allows to specify the problem instances much more economically, implying a reconsideration of the 
corresponding complexity measure that, as we prove, has important structural as well as algorithmic consequences.
\end{abstract}

%----------------------------------
\section{Introduction}

Cooperative game theory is a branch of game theory where a subset $X$ of players of an entirety $A$, forming a {\em coalition}, can achieve a common benefit. 
Unlike non-cooperative games, where each player is competing individually against the others, here the focus is cooperation.
We shall concentrate on {\em simple games} that specify which subsets $X$ of the collective $A$ are {\em winning coalitions};
i.e. can impose their common will, if they have one.
We will examine important properties that some simple games may have.
Given all its {\em minimal winning coalitions}, we will ask if the given game is {\em proper}, {\em strong}, {\em decisive}, {\em regular}, {\em linear}, {\em weighted}, {\em homogeneous} or {\em majority}.
So we will pose decision problems concerning simple games,
and we will try to characterize the computational complexity of these decision problems.

In \cite{FMOS11} the authors conjectured that given all the minimal winning coalitions of a simple game, 
the problem of deciding if this game is decisive ---or {\em self-dual}--- is {\em coNP}-complete. 
However, we will show that this {\em decisive problem} can be solved in quasi-polynomial time; that is, 
it can be solved by an algorithm which runs in a time bounded by a sub-exponential function, which is not necessarily polynomial.
In \cite{PS85} the authors proved that verifying that a game is weighted ---or {\em threshold} \cite{Mur71} or a {\em trade robust game} \cite{TZ92}--- is in $P$. 
But in \cite{FMOS11} the authors conjectured that deciding if a weighted game is decisive, is {\em coNP}-complete. 
We will prove that this conjecture is most probably wrong too (unless $NP=P$), because this problem is also in $P$. 
Actually, the decisive problem is polynomial even for linear games \cite{TZ99} ---also called {\em complete games} \cite{Car84}, {\em ordered games} \cite{Sud96}, {\em swap-robust games} \cite{TZ99} or {\em 2-monotone Boolean functions} \cite{Win62}--- which are simple games whose players can be ordered forming a regular game \cite{Mur71} ---also known as {\em directed game} \cite{KS95}---; a class that includes weighted games.

Finally ---going definitely beyond of what we already have stated in \cite{RP11}---
we shall drop the assumption that the given games are always specified by their minimal winning coalitions; and prove that regular games
can be given in a more economic way so that if we now measure the computational time according to this smaller input size, 
the mentioned algorithm of \cite{PS85} ceases to be polynomial. The corresponding decisive problem
turns out be at least as complex as for simple games.

In the following section, we introduce {\em hypergraphs} to review the main basic definitions.
In Section 3, we define the decisive problem and show that it is quasi-polynomial, using the well known result of \cite{FK96}.
In Section 4, we define regularity, weightedness and homogeneity, showing that one can recognize any of these properties in polynomial time;
and moreover, if one of them holds, also decide decisiveness in time polynomial in the input size. 
In Section 5, we concentrate on regular games, proving that for such games the minimal winning coalitions yield, in fact, an inflated measure of the input size.
In Section 6, we recapitulate considering {\em more than proper} games.

%----------------------------------
\section{Hypergraphs and simple games}

\begin{definition}
Given a finite {\em ground set} $A$ ---and its power set $\cP(A)$---
a {\em hypergraph} defined over $A$ is a family $\cH\seb\cP(A)$ of {\em hyperedges} $X\seb A$.

We may represent hypergraphs $\cH$ as incidence matrices, such that their rows represent the incidence vectors $x:A\rightarrow\{0,1\}$ of the hyperedges $X\in\cH$. So, given $a\in A$, $x(a)=1$ iff $a\in X$. The {\em size} of $\cH$, i.e., the amount of bits needed in order to write down the family of incidence vectors that characterize the hypergraph, is $|A|\cdot|\cH|\in\matN$.
\end{definition}

\begin{definition} \label{def2}
Given a hypergraph $\cH$, let be:
    \begin{itemize}
    \item $\neg(\cH):=\{A\bac X; X\in\cH\}$
    the family of complements of elements of $\cH$.
    \item $\mu(\cH):=\{X\in\cH;\fa Z\in\cH,Z\not\subset X\}$
    the family of {\em irredundant} ---or {\em minimal}--- elements of $\cH$.
    \item $\nu(\cH):=\{Z\seb A;\ex X\in\cH, X\seb Z\}$
    the family of subsets of $A$ that {\em respond} to $\cH$.
    \item $\tau(\cH):=\{Z\seb A;\fa X\in\cH, X\cap Z\not=\emptyset\}$
    the family of subsets of $A$ that are {\em transversal} to $\cH$.
    \item $\lambda(\cH):=\mu(\tau(\cH))$
    the family of {\em irredundant elements that are transversal} to $\cH$.
    \end{itemize}
\end{definition}

\begin{example}
The following six tables represent a hypergraph $\cH$ over $A:=\{3,2,1\}$, as well as the five families one obtains applying the defined operators.
In each table the first row represents the ground set $A$, and the rest of the rows present the incidence matrix that determines the denoted hypergraph.

\begin{flushleft}
\begin{minipage}[b]{14ex}
\begin{tabular}{|l|l|} \hline
$A$		& $321$\\ \hline
$\cH$		& $011$\\ 
		& $100$\\ 
		& $111$\\
		& $~$  \\
		& $~$  \\ \hline
\end{tabular}
\end{minipage}
\begin{minipage}[b]{15ex}
\begin{tabular}{|l|l|} \hline
$A$		& $321$\\ \hline
$\neg(\cH)$	& $000$\\ 
		& $011$\\ 
		& $100$\\
		& $~$  \\
		& $~$  \\ \hline
\end{tabular}
\end{minipage}
\hspace{1ex}
\begin{minipage}[b]{15ex}
\begin{tabular}{|l|l|} \hline
$A$		& $321$\\ \hline
$\mu(\cH)$	& $011$\\ 
		& $100$\\
		& $~$  \\
		& $~$  \\
		& $~$  \\ \hline
\end{tabular}
\end{minipage}
\hspace{1ex}
\begin{minipage}[b]{15ex}
\begin{tabular}{|l|l|} \hline
$A$		& $321$\\ \hline
$\nu(\cH)$	& $011$\\ 
		& $100$\\ 
		& $101$\\ 
		& $110$\\ 
		& $111$\\ \hline
\end{tabular}
\end{minipage}
\hspace{1ex}
\begin{minipage}[b]{15ex}
\begin{tabular}{|l|l|} \hline
$A$		& $321$\\ \hline
$\tau(\cH)$	& $101$\\ 
		& $110$\\ 
		& $111$\\
		& $~$  \\
		& $~$  \\ \hline
\end{tabular}
\end{minipage}
\hspace{1ex}
\begin{minipage}[b]{15ex}
\begin{tabular}{|l|l|} \hline
$A$		& $321$\\ \hline
$\lambda(\cH)$	& $101$\\ 
		& $110$\\
		& $~$  \\
		& $~$  \\
		& $~$  \\ \hline
\end{tabular}
\end{minipage}
\end{flushleft}
\end{example}

Let us recall some basic equations \cite{Pol08}:

\begin{lemma}\label{lema1}
Given $\cH\seb\cP(A)$, \, $\cP(A)\bac\nu(\cH)=\neg(\tau(\cH))$, \, $\tau(\nu(\cH))=\tau(\cH)$ \, and \, $\tau(\tau(\cH))=\nu(\cH)$.
\end{lemma}

\begin{proof}
The first equation follows directly from the definitions.
 The monotony of $\nu$ ---if $\cH\seb\cH'$, then $\nu(\cH)\seb\nu(\cH')$--- together with the antitony of $\tau$ ---if $\cH\seb\cH'$, then $\tau(\cH)\sepq\tau(\cH')$--- immediately imply $\tau(\nu(\cH))\seb\tau(\cH)$. Conversely,
 if $Z\in\tau(\cH)$, then $\fa X\in\cH$, $X\not\seb A\bac Z$; so $\fa Y\in\nu(\cH)$, $Y\not\seb A\bac Z$.  To prove the third equation, note that
 $\tau(\tau(\cH))=\{Y\seb A; \fa X\seb A, Y\seb A\bac X$ implies $X\not\in\tau(\cH)\}$; so the first equation implies 
 $\tau(\tau(\cH))=\{Y\seb A; \fa X\seb A, Y\seb A\bac X$ implies $A\bac X\in\nu(\cH)\}=\nu(\cH)$.
\end{proof}

Shortly after mathematical hypergraph theory was consolidated \cite{Ber70}, it was already used to analyze cooperative games \cite{Pol80}, 
although corresponding complexity considerations are quite new.

\begin{definition}
A {\em simple game} is a pair $\Gamma:=(A,\cW)$, where $A$ is a finite set of {\em players} and $\cW\seb\cP(A)$ is a family of {\em winning coalitions}, 
such that for all $X\in\cW$, $X\seb Z$ implies $Z\in\cW$. 
So $\cW$ is a hypergraph over $A$, with $\nu(\cW)=\cW$.
In other words, the Boolean function $f:\cP(A)\to\{0,1\}$ determined by the hypergraph $\cW$ ---such that $\cW=\{X\seb A; f(A)=1\}$--- is monotone.
Therefore we also say that simple games are monotone.

Then $\mu(\cW)$, the {\em irredundant kernel} of $\Gamma$, is the family of {\em minimal winning coalitions}; $\tau(\cW)$ is the {\em blocker} of $\cW$; and
$\cL:=\cP(A)\bac\cW=\{X\seb A; f(A)=0\}=\neg(\tau(\cW))$ is the set of {\em losing coalitions}. We do not demand $\emptyset\in\cL$, so $\cL$ can be void.
The simple game $(A,\tau(\cW))$ is called the {\em dual game} of $(A,\cW)$.
\end{definition}

%----------------------------------
\section{Decisiveness in simple games}

\begin{definition}\label{def4}
A simple game $(A,\cW)$ is:
    \begin{itemize}
    \item {\em proper}, if for all $S\seb A$, $S\in\cW \Rightarrow A\bac S\notin\cW$, i.e. $\cW\seb\tau(\cW)$,
    \item {\em strong}, if for all $S\seb A$, $S\notin\cW \Rightarrow A\bac S\in\cW$, i.e. $\cW\sepq\tau(\cW)$, and
    \item {\em decisive}, if it is both strong and proper, i.e. $\cW=\tau(\cW)$.
    \end{itemize}
\end{definition}

In \cite{TZ99} we read: {\em Properness rules out the possibility of disjoint winning coalitions, 
while strongness rules out the possibility of two losing coalitions whose union is $A$ $\ldots$ 
Some authors who view simple games as models of voting systems have little interest in simple games that are not proper. 
Their argument is that disjoint winning coalitions can allow contradictory decisions to be made by the voting body. $\ldots$ 
A less vigorous argument is sometimes raised against games that are not strong, and thus, the argument goes, 
leave some issues unresolved. Ramamurthy (1990) refers to ``the paralysis that may result from allowing a losing coalition to obstruct a decision.''}
So, ideally, games should be decisive; and this is why this property interests us so much.

\begin{definition} \cite{Pol08} A pair of hypergraphs $(\cH,\cK)$ over the same ground set is:
    \begin{itemize}
    \item {\em coherent}, if $\nu(\cH)\seb\tau(\cK)$, 
    \item {\em complete}, if $\nu(\cH)\sepq\tau(\cK)$, and
    \item {\em dual}, if it is both coherent and complete, i.e. $\nu(\cH)=\tau(\cK)$.
    \end{itemize}
\end{definition}

\begin{lemma}\label{lema2}
Given a simple game $\Gamma:=(A,\cW)$, then:
    \begin{itemize}
    \item $\Gamma$ is proper, iff $(\cW,\cW)$ is coherent.
    \item $\Gamma$ is strong, iff $(\cW,\cW)$ is complete.
    \item $\Gamma$ is decisive, iff $(\cW,\cW)$ is dual.
    \end{itemize}
\end{lemma}

\begin{proof}
Just note that:\\
$(\cW,\cW)$ is coherent iff $\nu(\cW)\seb\tau(\cW)$ iff $\cW\seb\tau(\cW)$, i.e. iff $\Gamma$ is proper.\\
$(\cW,\cW)$ is complete iff $\nu(\cW)\sepq\tau(\cW)$ iff $\cW\sepq\tau(\cW)$, i.e. iff $\Gamma$ is strong.\\
$(\cW,\cW)$ is dual iff $\nu(\cW)=\tau(\cW)$ iff $\cW=\tau(\cW)$, i.e. iff $\Gamma$ is decisive.
\end{proof}

Note that for any hypergraph $\cH$, $\nu(\mu(\cH))=\nu(\cH)$ and $\tau(\mu(\cH))=\tau(\cH)$.
So let $(\cH,\cK)$ be a pair of hypergraphs and $\Gamma:=(A,\cW)$ a simple game, then:

\begin{flushleft}
\begin{tabular}{l@{}l|l@{}l}
  $(\cH,\cK)$ is dual     & \, iff $(\mu(\cH),\mu(\cK))$ is dual     & $\Gamma:=(A,\cW)$ is decisive & \, iff $(\mu(\cW),\mu(\cW))$ is dual     \\
  $(\cH,\cK)$ is coherent & \, iff $(\mu(\cH),\mu(\cK))$ is coherent & $\Gamma:=(A,\cW)$ is proper   & \, iff $(\mu(\cW),\mu(\cW))$ is coherent \\
  $(\cH,\cK)$ is complete & \, iff $(\mu(\cH),\mu(\cK))$ is complete & $\Gamma:=(A,\cW)$ is strong   & \, iff $(\mu(\cW),\mu(\cW))$ is complete \\
\end{tabular}
\end{flushleft}

\begin{example}\label{ex2}
In projective geometry, the {\em Fano Plane} is the smallest projective plane. It was introduced in simple game theory by \cite{Ric56} to define a subclass of simple games called {\em finite projective games}. Since then, it has been very much studied, due to it has special properties that make it a likely counterexample for different results ---for instance, without going into details, it is the only non-partition game with the same number of minimal winning coalitions and players \cite{Sud96}--- as well as a case in which some properties turn out to be the same ---for instance, its reactive bargaining set coincides with its kernel \cite{GM97}.

It can be represented by a hypergraph $\cH$ over $A:=\{7,\ldots,1\}$, with seven evenly distributed hyperedges represented by the following incidence matrix. It is easy to check that $\mu(\cH)=\cH$ and prove that $(\cH,\cH)$ is coherent. It is more difficult to prove the completeness of this pair ---see the next Corollary \ref{cor1}--- but in fact it is dual. So the game $\Gamma:=(A,\nu(\cH))$ is proper and strong; i.e. decisive.

\begin{flushleft}
\begin{tabular}{|l|l|} \hline
$A$	& $7654321$\\ \hline
$\cH$	& $0000111$\\ 
	& $0011010$\\
	& $0101100$\\ 
	& $0110001$\\ 
	& $1001001$\\
	& $1010100$\\
	& $1100010$\\ \hline
\end{tabular}
\end{flushleft}
\end{example}

Given a simple game $\Gamma:=(A,\cW)$, since for any hypergraph $\cH$ we have $\mu(\nu(\cH))=\mu(\cH)$, 
we may assume that only $\mu(\cW)$ is extensively specified; that the rest of $\cW$ is only implicitly defined by $\mu(\cW)$. Because given any $Z\seb A$, 
deciding if $Z\in\cW$ can be done efficiently: Simply visit all $X\in\cH$, every time deciding $X\seb Z$; 
something that can be done in at most $|A|$ units of computation time.
So the amount of information needed to specify $\Gamma$ is given by the size $|A|\cdot|\mu(\cW)|$ of $\mu(\cW)$. 
This explains the following easy but important result:

\begin{corollary}\label{cor1}
For simple games $\Gamma:=(A,\cW)$,
\begin{enumerate}
    \item decide properness is polynomial.
    \item decide strongness is {\em coNP}-complete.
    \item decide decisiveness is quasi-polynomial; 
    i.e. there exists a decision algorithm whose logarithm of the running time is polynomial in the logarithm of the input size.
\end{enumerate}
\end{corollary}

\begin{proof}
The three statements rely on Lemma \ref{lema2}. Statement (1) follows because, as $\nu$ is a monotone operator, $(\cH,\cH)$ is coherent 
iff $\fa X,Z\in\cH$, $X\cap Z\not=\emptyset$, a condition that can be verified in polynomial time. Statement (2) follows because it has been proven many times \cite{PCPO02} that
the completeness of a pair of hypergraphs $(\cH,\cK)$, poses a decision problem that is {\em coNP}-complete. Moreover, this not necessarily symmetric problem can polynomially be reduced to the symmetric one~\cite{FK96}. Statement (3) follows from a deep result, proving that duality of a pair of hypergraphs $(\cH,\cK)$ can be decided in quasi-polynomial time, since there is a sub-exponential algorithm~\cite{FK96} whose logarithms of the running times can be bounded by a polynom of the logarithms of the sizes of the pairs. 
\end{proof}

Note that therefore the decisive problem is most probably not {\em NP}-hard, unless any {\em NP}-complete problem can be solved in quasi-polynomial time.
But note also that the mentioned quasi-polynomial algorithm does not allow us to generate $\lambda(\cH):=\mu(\tau(\cH))$ in sub-exponential time, 
since $|\lambda(\cH)|$ can not be quasi-polynomially bounded by $|\mu(\cH)|$. To prove this statement, consider the following example:

\begin{example}\label{ex3}
Let be $m\in\matN$, $I:=\{m,\ldots,1\}$, $n:=2m$, $A:=\{n,\ldots,1\}$ and $\cH:=\{\{2i,2i-1\}; i\in I\}$. 
Then, of course, $\lambda(\cH)=\{Y\seb A; \fa i\in I, 2i\in Y$ or (exclusive) $2i-1\in Y\}$.  So $|\cH|=m$ but $|\lambda(\cH)|=2^m$.
\end{example}

%----------------------------------
\section{Regular and weighted games}

\begin{definition}
Let $\leq$ denote a linear order on $A$ which orders the players by {\em increasing power}. 
Given $X\seb A$, an {\em increasing-shift} on $X$, specified by a pair $(a,b)\in A\bac X\times X$ such that $a>b$,
is an operation which returns $Z:=X\bac\{b\}\cup\{a\}$.

A simple game $(A,\cW)$ is {\em regular} with respect to the given linear order \cite{PS85,Mak02} if for all $X\in\cW$, every increasing-shift returns an element of $\cW$;
i.e. if replacing a member of a winning coalition $X\in\cW$ by a {\em more powerful} one, always yields a winning coalition.
\end{definition}

It is clear that a simple game $(A,\cW)$ is regular iff $\fa X\in\mu(\cW)$ every increasing-shift returns an element of $\cW$.
So, to decide the regularity of $\cW$, one can restrict the attention to $\cH:=\mu(\cW)$,
resulting that regularity can be decided in time polynomial in $|\cH|$.
In fact, this recognition problem can be solved in linear time \cite{Mak02}.

Regular games are ---in a sense that we shall reconsider in Section 5--- {\em very tractable}; i.e.
there exist algorithms \cite{PS94} that given $\cH:=|\mu(\cW)|$, 
generate $\lambda(\cH)$ in linear time; and at the same time prove that $|\lambda(\cH)|\leq |A|\cdot|\cH|+1$. 

\begin{definition}
A simple game $(A,\cW)$ is {\em linear} \cite{TZ99} ---or {\em $2$-monotonic} \cite{PS94}---  
if there exists a linear re-ordering of $A$, for which $\cW$ becomes regular.
\end{definition}

The linearity of $\cW$ can also be decided in time polynomial in $|\mu(\cW)|$;
since if one identifies $X\in\mu(\cW)$, $a\in A\bac X$ and $b\in X$ with $a>b$ and $\{a\}\cup X\bac\{b\}\not\in\nu(\cH)$, then, 
to have a chance to obtain a regular re-ordering, the order $a>b$ has to be reversed; definitively.
And such definitive re-ordering can be executed at most $n\cdot(n-1)/2$ times, where $n:=|A|$ \cite{Mak02}.\\

The Fano plane of Example \ref{ex2} is clearly not regular, and by its symmetry, not linear.

\begin{definition}
A {\em weighted game} is a simple game $\Gamma:=(A,\cW)$
that can be specified by a {\em threshold criterion} $(q,p)$ ---also known as {\em weighted representation}--- 
where $p:A\to\matR_+$ is a {\em weight function} and $q\in\matR_+$ a {\em quota} such that for all $Z\seb A$, $Z\in\cW$ iff $p(Z)\geq q$; 
where $p(Z):=\Sigma\{p(a); a\in Z\}$.
$\Gamma$ is {\em homogeneous} if it has a threshold criterion such that $\fa X\in\mu(\cW)$, $p(X)=q$.
\end{definition}

It is well known that each weighted game has a canonical weighted representation.

Weighted games were defined in 1944 by \cite{vNM44} and deeply studied in 1956 by \cite{Isb56}, in the context of simple game theory.
Since then they have been studied in many different contexts and under different names, like {\em linearly separated truth function} \cite{McNau61} 
---to contact and to rectify nets---, {\em linearly separable switching function} or {\em threshold Boolean functions} \cite{Hu65} 
---to separate circuits in switching circuit theory and analyse the threshold synthesis problem---, {\em trade robustness} \cite{TZ92} 
---for voting theory and trade exchanges--- or {\em threshold hypergraphs} \cite{Gol80,RRST85} ---to synchronize parallel processes. 
Homogeneous games were also first defined in \cite{vNM44}, and they are one of the most studied subclasses of weighted games \cite{Sud96}.

Weighted games clearly are simple (monotone) games, and one can always linearly re-order the elements of $A$ such that $p:A\to\matR$ becomes monotone; 
i.e. for all $a,b\in A$ with $a<b$, $p(a)\leq p(b)$. 
And then clearly the game $\cW$ becomes regular. So all weighted games are linear. The converse, however, does not necessarily hold \cite{TZ99,Mur71}.

According to the well known duality of linear inequality systems:

\begin{lemma} \label{lema3}
$\cH$ yields a weighted game $\nu(\cH)$ iff for all $u,u':\cH\to\matN$ with $u\cdot 1=u'\cdot 1$, $u\cdot H\not\leq u'\cdot H'$; where
$1$ is the $\cH$-vector with all entries equal to $1$, and $H,H'$ are the $\cH\times A$-incidence-matrices that represent $\cH$ and $\neg(\cH)$ respectively.
\end{lemma}

\begin{example} \label{ex4}
It is easy to verify that the simple game $(A,\cW)$ with $\cW:=\nu(\cH)$, given by the following hypergraph $\cH=\mu(\cW)$, is regular.
But according to Lemma \ref{lema3}, this game is not weighted: Just choose $u:\cH\to\matN$ as indicated in the following table, and $u':=u$.
In the next Section we will also prove that this game is decisive.

\begin{flushleft}
\begin{tabular}{|l|l|l|}   \hline
$A$			& $987654321$ & u\\\hline
$\cH$			& $011011011$ & 1\\
			& $011011101$ & 0\\
			& $011011110$ & 0\\
			& $011100100$ & 1\\
			& $011101000$ & 0\\
			& $011110000$ & 0\\
			& $100011100$ & 1\\
			& $100100011$ & 1\\
			& $100100101$ & 0\\
			& $100100110$ & 0\\
			& $100101000$ & 0\\
			& $100110000$ & 0\\
			& $101000000$ & 0\\
			& $110000000$ & 0\\\hline
$\Sigma(z\cdot u)$	& $222222222$ & 4\\\hline
\end{tabular}
\end{flushleft}
\end{example}
$~$\\

The next result was proven in \cite{PS85}, but not for homogeneous games.

\begin{theorem} \label{teo1}
Deciding whether a given simple game $\Gamma:=(A,\cW)$ is weighted (homogeneous) or not, can be done in time polynomial in the size of $\cH:=\mu(\cW)$.
\end{theorem}

\begin{proof}
The following procedure is efficient for both kind of games, since each Step can be accomplished in polynomial time: 
The first four were already commented above, and Steps 5 to 8 only demand the solution of a system of linear equalities and inequalities, 
which can also be found in polynomial time \cite{Kha79}. Note that in Step 8 the original restriction demands that for each loser $Y\in\cJ$, $p(Y)<q$; 
but this can equivalently be replaced by $\fa Y\in\cJ$, $p(Y)\leq q-1$, because multiplying a threshold criterion by a positive constant, returns a threshold criterion.
\end{proof}

\begin{algorithm}
\begin{algorithmic}[1]
\REQUIRE The irredundant kernel $\cH:=\mu(\cW)$ of a simple game $\Gamma:=(A,\cW)$.
\ENSURE  If $\Gamma$ is weighted (homogeneous), return ``Yes''; otherwise return ``No''.
\STATE Decide if $\cH$ is linear;
\STATE \textbf{If} this is not the case, \textbf{then} return ``No'';
\STATE Determine a linear ordering of $A$ which makes $\cH$ regular;
\STATE Generate $\cJ:=\neg(\lambda(\cH))$;
\STATE Decide the existence of a threshold criterion $(q,p)$, such that:\\
$\fa a\in A$, \,\,\,\,\, $p(a)\geq 0$;\\
$\fa X\in\cH$, \, $p(X)\geq q$ ~~($p(X)=q$);\\
$\fa Y\in\cJ$, \,  $p(Y)\leq q-1$;
\STATE \textbf{If} this is not the case, \textbf{then} return ``No'';
\STATE \textbf{Return} ``Yes'' --- $\Gamma$ is weighted (homogeneous). 
\end{algorithmic}
\caption{IsWeighted}\label{alg1}
\end{algorithm}

Next we will prove a similar result, but now regarding decisive games:

\begin{definition}
A weighted game $\Gamma:=(A,\cW)$ is a {\em sub-majority game}, if it is strong; and it is a {\em majority game}, if it is decisive.
\end{definition}

In \cite{FMOS11} the authors conjecture that verifying majority games is {\em coNP}-complete.
This conjecture is most probably false (unless $P$=$NP$), because:

\begin{theorem} \label{teo2}
Deciding whether a given simple game $\Gamma:=(A,\cW)$ is majority (sub-majority) or not, can be done in time polynomial in the size of $\cH:=\mu(\cW)$.
\end{theorem}

\begin{proof}
Analogously to Theorem \ref{teo1}, the following Algorithm \ref{alg2} is efficient to decide both kind of games.
First note that by Definition \ref{def4}, $\Gamma$ is decisive iff $\tau(\cW)=\cW$, so applying the operator $\mu$ on the expression we obtain $\lambda(\cW)=\mu(\cW)$, by Definition \ref{def2}. In the same way, $\Gamma$ is strong iff $\tau(\cW)\seb\cW$ iff $\lambda(\cW)\seb\mu(\cW)$.
Moreover, it is easy to see that $\lambda(\mu(\cW))=\lambda(\cW)$.
Using these ideas, we test decisiveness or strongness on Steps 4-5, through the following questions: is $\lambda(\cH)=\cH$? or, respectively, is $\lambda(\cH)\seb\nu(\cH)$? which are equivalent to the two conditions given above. Finally, the procedure to decide whether the game is also weighted or homogeneous is the same as the given by Algorithm~\ref{alg1}.
\end{proof}

\begin{algorithm}
\begin{algorithmic}[1]
\REQUIRE The irredundant kernel $\cH:=\mu(\cW)$ of a simple game $\Gamma:=(A,\cW)$.
\ENSURE  If $\Gamma$ is majority (sub-majority), return ``Yes''; otherwise return ``No''.
\STATE Decide if $\cH$ is linear;
\STATE \textbf{If} this is not the case, \textbf{then} return ``No'';
\STATE Generate $\cK:=\lambda(\cH)$;
\STATE Decide if $\cK=\cH$ ~~($\cK\seb\nu(\cH)$);
\STATE \textbf{If} this is not the case, \textbf{then} return ``No'';
\STATE Decide the existence of a threshold criterion $(q,p)$, such that:\\
$\fa a\in A$, \,\,\,\,\,\, $p(a)\geq 0$;\\
$\fa X\in\cH$, \,\, $p(X)\geq q$;\\
$\fa Y\in\neg(\cK)$, $p(Y)\leq q-1$;
\STATE \textbf{If} this is not the case, \textbf{then} return ``No'';
\STATE \textbf{Return} ``Yes'' --- $\Gamma$ is majority (sub-majority).
\caption{IsMajority}\label{alg2}
\end{algorithmic}
\end{algorithm}

It is known that up to eight players, all the simple games which are regular and decisive are also majority. 
This implies that Step 6 of Algorithm \ref{alg2} is unnecessary in these cases, because the answer would always be ``Yes''. 
However, from nine players onwards, this is not always true. It is known that for $|A| = 9$ there are $319124$ regular and decisive games \cite{KS95}, but only $175428$ majority games \cite{MTB70}. 

As far as we know, explicit examples of regular games which are decisive but not weighted, as our Example \ref{ex4}, have not been indicated until now. 
However, another regular game that in fact meets these characteristics is given in \cite{Sud96b}.

%----------------------------------
\section{Economic specification of regular games}

The considerations that in the last section led from simple games to weighted ones, revealed the central importance of regular games.
Mainly because to prove Theorem \ref{teo1} and Theorem \ref{teo2}, we had to rely on the main result of~\cite{PS85}, which,
given the minimal winning coalitions $\mu(\cW)$ of a regular game $\Gamma:=(A,\cW)$, allows to generate the family of the irredundant transversal $\lambda(\cW)$
in polynomial time, in function of the input size. But there is, as we will see, a much more economic way to determine regular games;
one that will demand a reconsideration of the recalled polynomiallity results.\\

Since in what follows we restrict the attention to regular games, we may assume that $A=\{n,n-1,\ldots,1\}$.

\begin{definition}
Given $X\seb A$, a {\em left-shift} on $X$ is either an increasing-shift ---see Definition 6--- specified by a pair $(a,b)\in A\bac X\times X$, such that $a=b+1$; or, if $1\in A\bac X$, a replacement of $X$ by $X\cup\{1\}$.

Given $X,Z\seb A$ we say that $Z$ is a {\em successor} of $X$, if a left-shift performed on $X$ yields $Z$;
i.e. if there exist $a\in A\bac X$ and $b\in X\cup\{0\}$, such that $a=b+1$ and $Z:=X\bac\{b\}\cup\{a\}$.
So we can conceive a directed {\em successor graph} on the {\em set of nodes} $\cP(A)$ such that, given $X,Z\seb A$, $(X,Z)$ is an {\em arrow} of the graph, if $Z$ is a successor of $X$.

We write $X\seb'Z$ if, starting with $X$, there exists a sequence of left-shifts that produces $Z$.
It is quite clear that a sequence of left-shifts starting with $X$ can not re-produce $X$ ---unless the sequence is empty---
so the successor graph is acyclic. The relation $\seb'$ is the reflexive-and-transitive closure of that successor graph. So it is an order relation,
that we shall call {\em shift order}.
\end{definition}

Given $X,Z\seb A$, $X\seb'Z$ iff for all $a\in A$, $|\{b\in X; b\geq a\}|\leq |\{b\in Z; b\geq a\}|$;
because if $Z$ is a successor of $X$, then the stated inequalities evidently hold; and, the other way round, if the inequalities hold, then starting from $X$, one can always perform a sequence of left-shifts that preserve the inequalities to finally produce $Z$.
So this shift relation can be decided in time polynomial in $n$.
And it is also clear that, if $X\seb Z$, then $X\seb'Z$; i.e. that $\seb'$ is a monotone variation of $\seb$. 
And that for $\seb'$ ---as for $\seb$--- complementation is antitone; i.e. if $X\seb'Z$, then $A\bac Z\seb'A\bac X$.

The operators $\nu',\tau',\mu',\lambda'$ are defined like the already familiar $\nu,\tau,\mu,\lambda$ ---see Definition \ref{def2}--- but with the relation $\seb'$ instead of $\seb$. For instance, given $\cH\seb\cP(A)$, $\tau'(\cH):=\{Z\seb A;\fa X\in\cH, X\not\seb'A\bac Z\}$, and 
$\mu'(\cW):=\{X\in\cW; \fa Z\in\cW, Z\not\subset' X\}$ yields the {\em shift-minimal winning coalitions} of the given game.

Note that, given a hypergraph $\cW$ over $A$, $(A,\cW)$ is a regular game iff $\nu'(\cW)=\cW$. This is so because we always have
$\nu'(\cW)\sepq \nu(\cW)\sepq \cW$. So $\nu'(\cW)= \nu(\cW)= \cW$ is equivalent to $\nu'(\cW)\seb \cW$. Hence regular games are $\seb'$-monotone games,
as simple games are $\seb$-monotone games.

Some of these notions and results are known since \cite{Ost87}.

\begin{lemma}
 A simple game $\Gamma:=(A,\cW)$ is regular iff $\cW=\nu'(\cH')$, where $\cH':=\mu'(\cH)$.\\
 If $\Gamma$ is regular, then $\tau(\cW)=\tau'(\cW)$, so its dual game $(A,\tau(\cW))$ is regular too. Therefore\\
 a regular $\Gamma$ is proper iff $\nu'(\cW)\seb\tau'(\cW)$; is strong iff $\nu'(\cW)\sepq\tau'(\cW)$; and is decisive iff $\nu'(\cW)=\tau'(\cW)$.
\end{lemma}

\begin{proof}
 The above mentioned monotone variation, together with the definition of $\cH'$, imply $\nu(\cW)\seb\nu'(\cW)=\nu'(\cH')$. 
 So, to prove the first statement, it only remains to point out that regularity of $(A,\cW)$ is equivalent to $\nu'(\cW)\seb\nu(\cW)$.
 The second statement follows from the variation of Lemma \ref{lema1} that one obtains replacing $\seb$ by $\seb'$. 
 Note that the given proof of Lemma \ref{lema1} holds for any order relation on $\cP(A)$ for which complementation is antitone.
\end{proof}

Hence, to implicitly specify a regular game $(A,\cW)$, it suffices to explicitly specify $\cH':=\mu'(\cW)\seb\cH:=\mu(\cW)$:
Given any $Z\seb A$, one can decide $Z\in\cW$, simply visiting all $X\in\cH'$, every time deciding $X\seb' Z$.

\begin{example}\label{ex5}
Consider the regular game of Example \ref{ex4}, specified over there by a hypergraph $\cH$.\\
Its economic specification $\cH'$ is the following; because all other $X\in\cH$ are such that $X\in\nu'(\cH')$.

\begin{flushleft}
\begin{tabular}{|l|l|l|}  \hline
$A$			& $987654321$ & u\\\hline
$\cH'$			& $011011011$ & 1\\
			& $011100100$ & 1\\
			& $100011100$ & 1\\
			& $100100011$ & 1\\
			& $100101000$ & 0\\
			& $101000000$ & 0\\\hline
\end{tabular}
\end{flushleft}
\end{example}
$~$\\

Until now we have assumed that not only simple games $\Gamma:=(A,\cW)$ but also, in particular, regular games, will be characterized by $\mu(\cW)$;
and that therefore the time needed to decide questions concerning simple games,
should be measured using the {\em simple-measure} $\kappa(\Gamma):=|A|\cdot|\mu(\cW)|$; the size of $\mu(\cW)$.
But if we limit the attention to the universe of regular games $\Gamma:=(A,\cW)$, since we now know that they can also be characterized by $\mu'(\cW)$,
then it seems to be reasonable that the deciding times should be measured using the {\em regular-measure}
$\kappa'(\Gamma):=|A|\cdot|\mu'(\cW)|$; the size of $\mu'(\cW)$.
One reason for this, is that the characterization in question turns out to be efficient,
since for all $Z\seb A$, $Z\in\cW$ can be decided in polynomial time in function of $\kappa'(\Gamma)$.
Another reason is, that
for all regular games $\Gamma:=(A,\cW)$, since $\mu'(\cW)\seb\mu(\cW)$, also $\kappa'(\Gamma)\leq\kappa(\Gamma)$.

But in our opinion this adaptation to the universe of regular games is moreover mandatory, because $\kappa$ would yield an {\em inflated measure} for regular games, 
compared with $\kappa'$,
since there exist sequences $\{\Gamma_m; m\in\matN\}$ of regular games, 
such that $\kappa(\Gamma_m)/\kappa'(\Gamma_m)$ grows exponentially in function of $m\in\matN$.
For instance, if we define $\Gamma_m$ such that $A:=\{2m,\ldots,1\}$, $X:=\{a\in A; m\geq a\}$,
and $\cW:=\nu'(\{X\})$, then $|\mu'(\cW)|=1$ and $|\mu(W)|={2m\choose m}$.

This opinion contradicts the one of the authors of~\cite{PS85}. They proved that there exists a polynom $p:\matN\to\matN$, such that for all $n\in\matN$,
$\hat{\kappa}(n)/\hat{\kappa}'(n)\leq p(n)$; where $\hat{\kappa}(n):=\max\{\kappa((A,\cW)); |A|=n\}$ and $\hat{\kappa}'(n):=\max\{\kappa'((A,\cW)); |A|=n\}$.
They even proved that the same would still hold, if $\kappa'$ would be replaced by the size of any codification which is able to distinguish different regular games.
And based on this result, they conclude that the simple-measure is a {\em not-inflated measure}.

This view, however, forces us to ``scan'' the universe of regular games partitioning it in ``slices'' with constant $|A|$. But one could
as well ``scan'' it, partitioning it in ``slices'' of constant $\kappa'(\cW)$. Then, redefining $\hat{\kappa},\hat{\kappa}'$ accordingly,
for any $m\in\matN$ one would have $\hat{\kappa}'(2m)=2m$. But since the above defined regular game
$\Gamma_m$ is in the ``slice'' $2m$, $\hat{\kappa}(2m)\geq 2m\cdot{2m\choose m}$. 
So $\hat{\kappa}(2m)/\hat{\kappa}'(2m)$ can not be bounded by a polynom in $m$.

For all these reasons we adopt the regular-measure $\kappa'$; with all its important consequences:\\

\begin{definition}
Let $A:=\{n,\ldots,1\}$ and $A':=\{2n,2n-1,\ldots,2,1\}$.\\
Let $\cP(A)$ be ordered according to $\seb$ and $\cP(A')$ be ordered according to $\seb'$.\\
Let $T:\cP(A)\to\cP(A')$ be such that given $X\seb A$ and $a\in A$, $2a\in T(X)$ iff $a\in X$ and $2a-1\in T(X)$ iff $a\not\in X$.\\
So $T$ is a monomorphism, because it is an injective function such that for all $X,Z\seb A$, $X\seb Z$ iff $T(X)\seb' T(Z)$
---since $T(Z)$ is a successor of $T(X)$ iff $Z$ is a set-theoretical successor, i.e. a superset, of $X$.\\
For any $a\in A$, let $Z^a:=\{2a\}\cup\{\{2b-1; b\in A, b\geq a\}\seb A'$ ---see Example \ref{ex6}--- and let $\cG':=\{Z^a; a\in A\}\seb\cP(A')$.
\end{definition}

\begin{lemma} \label{lema-reduction}
The image $T(\cP(A)):=\{T(X);X\seb A\}$ of $\cP(A)$  is closed with respect to complementation.\\
Restricting the codomain to $T(\cP(A))$, $T$ becomes bijective; i.e. an isomorphism.\\
Moreover $\nu'(\cG')\cup T(\cP(A))=\tau'(\cG')$.
\end{lemma}

\begin{proof}
The $T$-image of $\cP(A)$ covers (only) 
the $X'\in\cP(A')$ that for all $a\in A$, either satisfy $2a\in X'$ or (exclusive) $2a-1\in X'$. So it is closed with respect to complementation.
Let us now prove the last statement. We have $\nu'(\cG')\cup T(\cP(A))\seb\tau'(\cG')$, because if $X'\in\nu'(\cG')\cup T(\cP(A))$ and $a\in A$,
then $X'\seb' A'\bac Z^a$ is impossible, since $|\{b'\in X'; b'\geq 2a-1\}|> n-a=|\{b'\in A'\bac Z^a; b'\geq 2a-1\}|$. 
And we have $\nu'(\cG')\cup T(\cP(A))\sepq\tau'(\cG')$, because this inclusion is equivalent to $\cP(A')\bac T(\cP(A))\seb\nu'(\cG')\cup\nu'(\neg(\cG'))$. And,
given $Z'\in\cP(A')\bac T(\cP(A))$, let $a\in A$ be maximal such that either $\{2a,2a-1\}\seb Z'$ or $\{2a,2a-1\}\cap Z'=\emptyset$.
Then if the first holds, clearly $Z^a\seb' Z'$; and if the second one holds, $A'\bac Z^a\seb' Z'$.
\end{proof}

\begin{lemma} \label{lem6}
  Given $\cH,\cK\seb\cP(A)$, $\cH':=\{T(X);X\in\cH\}$ and $\cK':=\{T(Y);Y\in\cK\}$.\\
 Then $\nu(\cH)\seb\tau(\cK)$ iff $\nu'(\cH'\cup\cG')\seb\tau'(\cK'\cup\cG')$, 
 so the pair $(\cH,\cK)$ is coherent iff $(\cH'\cup\cG',\cK'\cup\cG')$ is coherent;\\
 and $\nu(\cH)\sepq\tau(\cK)$ iff $\nu'(\cH'\cup\cG')\sepq\tau'(\cK'\cup\cG')$, 
 so the pair $(\cH,\cK)$ is complete iff $(\cH'\cup\cG',\cK'\cup\cG')$ is complete.
\end{lemma}

\begin{proof} 
Clearly $\nu(\cH)\seb\tau(\cK)$ iff $\fa (X,Y)\in\cH\times\cK$, $X\not\seb A\bac Y$; i.e. iff $\fa (X',Y')\in\cH'\times\cK'$, $X'\not\seb'A'\bac Y'$;
i.e. iff $\nu'(\cH')\seb\tau'(\cK')$. Moreover our Lemma 5 implies $\nu'(\cG')\seb\tau'(\cG')$, $\nu'(\cH')\seb\tau'(\cG')$ and $\nu'(\cG')\seb\tau'(\cK')$.
So we completed the proof of the first statement.
On the other hand, $\nu(\cH)\sepq\tau(\cK)$ iff $\fa Z\in\cP(A)$ with $A\bac Z\not\in\nu(\cK)$, $Z\in\nu(\cH)$; i.e. iff
$\fa Z'\in T(\cP(A))$ with $A'\bac Z'\not\in\nu'(\cK')$, $Z\in\nu(\cH')$. And since $\fa Z'\in\cP(A')\bac T(\cP(A))$ with $A'\bac Z'\not\in\nu'(\cG')$, $Z\in\nu(\cG')$,
we have proved the second statement. 
\end{proof}

From this Lemma, setting $\cH:=\cK:=\cW$, and considering that $|\cG'|\leq |A|$, we immediately get the following:

\begin{theorem} \label{teo3}
 The decisiveness (resp. strongness) problem for simple games $(A,\cW)$ can polynomially be reduced to the
decisiveness (resp. strongness) problem for regular games $(A',\cW')$, specified by their shift-minimal winning coalitions $\mu'(\cW')$.
This implies that the strongness problem for regular games is {\em coNP}-complete.
And that if the decisive problem for regular games $(A',\cW')$ turns out to be polynomial in $|A'|\cdot|\mu'(\cW')|$,
then the decisive problem for simple games $(A,\cW)$ will be polynomial in $|A|\cdot|\mu(\cW)|$.
\end{theorem}

And remembering that $T(\cP(A))\cap\nu'(\cG')=\emptyset$, we immediately get the following:

\begin{corollary}\label{cor2}
 Given $\cH\seb\cP(A)$, let $\cK:=\lambda(\cH)$, $\cH':=\{T(X);X\in\cH\}$ and $\cK':=\{T(Y);Y\in\cK\}$.\\
 Then $\cK'\seb\lambda'(\cH'\cup\cG')=\mu'(\cK'\cup\cG')$. 
 Therefore $\lambda(\cH)$ can be deduced from $\lambda'(\cH'\cup\cG')$.
\end{corollary}

Let us now apply our Corollary \ref{cor2}. Reconsider Example \ref{ex3} at the end of Section 3, where
$I:=\{m,\ldots,1\}$, $A:=\{2m,2m-1,\ldots,2,1\}$, $\cH:=\{\{2i,2i-1\}; i\in I\}$ 
and $\cK=\{Y\seb A; \fa i\in I, 2i\in Y$ or (exclusive) $2i-1\in Y\}$.
Then $|\cG'|=2m$, so $|\cH'\cup\cG'|=3m$. Therefore, with $k:=3m$,
$|\cH'\cup\cG'|=k$ and $|\lambda'(\cH'\cup\cG')|\geq |\cK'|=c^{k}$, where $c:=2^{1/3}>1$;
so the size of the hypergraph $\lambda'(\cH'\cup\cG')$ grows exponentially in the size of $\mu'(\cH'\cup\cG')$.

\begin{example}\label{ex6}
The following table presents the case $m:=2$.

\begin{flushleft}
\begin{minipage}[b]{15ex}
\begin{tabular}{|l|l|}  \hline
$A$			& $4321$\\ \hline
$\cH$			& $1100$\\
			& $0011$\\ \hline 
$\cK$			& $1010$\\
			& $1001$\\
			& $0110$\\
			& $0101$\\ 
			& $~$   \\
			& $~$   \\
			& $~$   \\
			& $~$   \\ \hline
\end{tabular}
\end{minipage}
\hspace{1ex}
\begin{minipage}[b]{15ex}
\begin{tabular}{|l|l|}  \hline
$A'$			& $87654321$\\\hline
$\cH'$			& $10100101$\\
			& $01011010$\\ \hline 
$\cK'$			& $10011001$\\
			& $10010110$\\
			& $01101001$\\
			& $01100110$\\\hline
	$Z^4$		& $11000000$\\
	$Z^3$		& $01110000$\\
	$Z^2$		& $01011100$\\
	$Z^1$		& $01010111$\\ \hline
\end{tabular}
\end{minipage}
\end{flushleft}
\end{example}
$~$\\

This also proves that for regular games $(A,\cW)$, $|\lambda'(\cW)|$ can grow exponentially in function of $|\mu'(\cW)|$;
although $|\lambda(\cW)|$ can be bounded linearly in function of $|\mu(\cW)|$~\cite{PS85}.
This result was already proven before, using different techniques, in Corollary 4 of~\cite{KKZ13}.

Our results translate the ``bad news'' of classical duality theory~\cite{FK96} to corresponding ``bad news'' for the regular duality theory.
But our polynomial-reduction-results would also guaranty that ``good news'' for the regular cases translate to corresponding ``good news'' for the classical theory.

If we have a good algorithm to determine $\lambda'(\cH')$ from $\cH':=\mu'(\cW)$, then using that algorithm is certainly preferred to using the Hop-Skip-and-Jump algorithm of~\cite{PS85}.
Suppose that for any regular game $\Gamma:=(A,\cW)$, starting with $\cH':=\mu'(\cW)$ we could determined $\cK':=\lambda'(\cW)$ in time polynomially bounded by the {\em input plus output size} $|A|\cdot(|\cH'|+|\cK'|)$.
This does not contradict our Corollary \ref{cor2}; but according to our Lemma \ref{lem6} it would imply {\em sensational good news} for classical duality theory: It would prove that all the so much investigated decision problems, that until now are only known to be quasi-polynomial~\cite{FK96}, are in fact polynomial.

So, summing up: The here developed critical considerations do not question Theorem 1 and 2, because they refer to simple games.
But if we would restrict the considered problems to regular games, then we can not solve them polynomially any more.
We do not even know if the decisive problem for regular games $\Gamma:=(A,\cW)$ can be solved in quasi-polynomial time, 
in function of $\kappa'(\Gamma):=|A|\cdot|\mu'(\cW)|$.

%----------------------------------
\section{Conclusions}

As we have seen, strongness poses the main question: Given $\cW\seb\cP(A)$, is $\cW\cup\neg(\cW)=\cP(A)$?

If $\cW$ is given explicitly, then obviously the question can be answered in time, that can polynomially be bounded by the size $|A|\cdot |\cW|$ of the  given $\cW$.
But if $\cW$ is monotone, then $\cW$ can implicitly be given by $\cH:=\mu(\cW)$.
So, if we restrict the attention to simple games $(A,\cW)$, as we did, then the size of the problems instances is $|A|\cdot|\cH|$. 
And evaluating according to that unit of measurement, the strongness problem becomes {\em coNP}-complete; 
unless we restrict the scope even more, to proper games, for which $\cW\seb\tau(\cH)$ holds.
Then strongness can be decided in time that can be bounded by a sub-exponential function of the size of $\cH$. 
But the corresponding quasi-polynomial algorithms are definitely super-polynomial.
To guaranty polynomiallity, it seems that one has to restrict the universe to instances $\cW$ that are ``more-than-proper''. 
For example, if $A=\{n,\ldots,1\}$ and for all considered instances $\nu'(\cH)\seb\tau'(\cH)$ holds, 
then $\nu(\cH)\seb\nu'(\cH)\seb\tau'(\cH)\seb\tau(\cH)$; so strongness $\nu(\cH)\sepq\tau(\cH)$ implies regularity. 
Hence, to decide strongness, one can first examine regularity and, if the answer is positive, decide the strongness of the regular game. 
All, of course, in time polynomial in $|A|\cdot|\cH|$. 
In Section 4 we derived a similar result for weighted instances; which in fact are dually characterized by the, properness implying, 
``more than proper'' condition of Lemma \ref{lema3}.\\

However, weighted instances $\cW$ will often not be given stating explicitly $\cH:=\mu(\cW)$. Usually they will be defined implicitly, 
explicitly giving only a corresponding threshold criterion $(p,q)$. Then, measuring computational time according to the size of $(p,q)$ 
---which of course is much smaller than $|A|\cdot|\cH|$--- the strongness problem turns out to be {\em coNP}-complete; as is well known \cite{FMOS11}. 
Again we could focus on a sub-universe: Suppose that $A$ can be reordered such that $A=\{n,\ldots,1\}$ and for all $a\in A$, $p(a)=2^a$. 
Then there exists exactly one $X\in\cW$, such that $p(X):=\Sigma\{p(a); a\in A\}$ is minimal. So, to implicitly give $\cW$, 
it suffices to explicitly exhibit this minimal winning coalition $X\in\cW$.
The size of this $X$ is $|A|$: very small. Nonetheless, strongness can be decided in time polynomial in that small size: 
Simply determine the maximal loser $Y\seb A$ ---the unique one, such that $p(X)-1=p(Y)$--- and decide if $A\bac Y\in\cW$.
These considerations show that restricting our strongness problem to ``more than proper'' sub-universes of simple games, sometimes, but not always, 
really reduces the problems complexity.

In Section 5 we proved that if instead of considering simple games we restrict the attention to the regular ones, it does not reduce the problems complexity; 
because regular games can implicitly be specified by codes that are much smaller than the ones needed to specify simple games.
This turned out to be the reason why ---in Section 5--- the strongness decision problem for simple games could be polynomially reduced to its restriction to regular games.
And the same holds if from the beginning we restrict the attention to proper games. These results have, as we indicated, algorithmic implications.

Beyond that, they are also of interest for structural game theoretical considerations~\cite{Pol80,Pol08}. 
Because our last Theorem \ref{teo3} proves that one can extend the monomorphism $T$ defined in Section 5, 
to define for every simple game $\Gamma:=(A,\cW)$ a regular game $\Gamma':=T(\Gamma):=(A',T(\cW))$, 
such that the amount of information needed to specify $T(\Gamma)$ is polynomial in the amount required by $\Gamma$. In this sense we may conclude that
the class of regular, i.e. $\seb'$-monotone games, yields a structure that is {\em at least as eloquent} as the class
of simple, i.e. $\seb$-monotone games; because this classical {\em Boolean structure} can
economically be interpreted in the here reconsidered {\em regular structure}.

%----------------------------------

\bibliographystyle{model1-num-names}

\end{document}